\documentclass[11pt,a4paper]{article}
\usepackage{geometry,color}

\usepackage[latin1]{inputenc}

\usepackage{amsmath,amsthm,amssymb,subdepth}
\usepackage{amsfonts,mathrsfs,enumerate}
\usepackage{setspace}
\usepackage{hyperref}
\usepackage{url}
\title{A note on mass-minimising extensions}
\date{\today}
\author{Stephen\hspace{-1.5mm}\newlength{\Mheight}
\newlength{\cwidth}
\settoheight{\Mheight}{M}\settowidth{\cwidth}{c}M\parbox[b][\Mheight][t]{\cwidth}{c}\hspace{0.15mm}Cormick\footnote{stephen.mccormick@une.edu.au}\vspace{3mm}\\School of Science and Technology\\University of New England\\Armidale, 2351\\Australia}
\newtheorem{theorem}{Theorem}[section]

\newtheorem{proposition}[theorem]{Proposition}

\newtheorem{conjecture}[theorem]{Conjecture}

\newtheorem{lemma}[theorem]{Lemma}

\newtheorem{remark}[theorem]{Remark}

\newcommand{\onabla}{\mathring\nabla}
\newcommand{\tnabla}{\tilde{\nabla}}
\newcommand{\tg}{\tilde{g}}

\newcommand{\dn}{{\delta_0}}
\newcommand{\og}{\mathring g}

\newcommand{\R}{\mathbb{R}}

\newcommand{\opi}{\mathring{\pi}}
\newcommand{\W}{\overline{W}}
\newcommand{\bH}{\overline{H}}

\numberwithin{equation}{section}

\usepackage{fancyhdr}
\begin{document}

\title{A note on mass-minimising extensions
}

\author{Stephen\hspace{-1.5mm}\newlength{\Mheight}
\newlength{\cwidth}
\settoheight{\Mheight}{M}\settowidth{\cwidth}{c}M\parbox[b][\Mheight][t]{\cwidth}{c}\hspace{0.15mm}Cormick\footnote{stephen.mccormick@une.edu.au}\vspace{3mm}\\School of Science and Technology\\University of New England\\Armidale, 2351\\\vspace{3mm}Australia\\ \today}

\date{\small Accepted for publication in General Relativity and Gravitation \\ The final publication is available at Springer via \href{http://dx.doi.org/10.1007/s10714-015-1993-2}{http://dx.doi.org/10.1007/s10714-015-1993-2}}

\setstretch{1.15}

\maketitle

\begin{abstract}
A conjecture related to the Bartnik quasilocal mass, is that the infimum of the ADM energy, over an appropriate space of extensions to a compact 3-manifold with boundary, is realised by a static metric. It was shown by Corvino [Comm. Math. Phys. 214(1), (2000)] that if the infimum is indeed achieved, then it is achieved by a static metric; however, the more difficult question of whether or not the infimum is achieved, is still an open problem. Bartnik [Comm. Anal. Geom. 13(5), (2005)] then proved that critical points of the ADM mass, over the space of solutions to the Einstein constraints on an asymptotically flat manifold without boundary, correspond to stationary solutions. In that article, he stated that it should be possible to use a similar construction to provide a more natural proof of Corvino's result.

In the first part of this note, we discuss the required modifications to Bartnik's argument to adapt it to include a boundary. Assuming that certain results concerning a Hilbert manifold structure for the space of solutions carry over to the case considered here, we then demonstrate how Bartnik's proof can be modified to consider the simpler case of scalar-flat extensions and obtain Corvino's result.

In the second part of this note, we consider a space of extensions in a fixed conformal class. Sufficient conditions are given to ensure that the infimum is realised within this class.

\end{abstract}

\section{Introduction}
The Bartnik mass is often cited as the gold standard definition of a quasilocal mass\footnote{We borrow the phrase ``gold standard" from a quote by Hubert Bray in a Duke University press release.}, if only it were possible to compute for a generic domain. The mass of a domain $\Omega$ in some initial data 3-manifold is as taken to be the infimum of the ADM mass over a space of admissible extensions to $\Omega$, satisfying the Einstein constraints. In Ref. \cite{qlm}, where Bartnik first defined the quasilocal mass now bearing his name, it is conjectured that this infimum is realised by a static extension to $\Omega$.

In 2000, Corvino proved part of this conjecture (Theorem 8 of Ref. \cite{corvino2000}); he proved that if a minimal ADM energy extension exists then it must be static. Note that we differentiate between the energy and the mass -- the latter being the absolute value of the energy-momentum four-vector, while the former refers to the component that is orthogonal to the Cauchy surface. It was then shown by Miao \cite{pengzi2003} that this static extension must also satisfy Bartnik's geometric boundary conditions; that is, the metric and boundary mean curvature agree on either side of $\partial\Omega$. Later, Bartnik suggested that a variational proof of Corvino's result, based on extending his work on the phase space \cite{phasespace} to manifolds with boundary, would be more natural.

In the first part of this note, we discuss how Bartnik's analysis may be modified to the case where the data is fixed on the boundary to provide an alternate proof of Corvino's result. The extensions considered here fix the first derivative of the metric on the boundary, which is a stronger condition than the usual Bartnik data. In the context of the Bartnik conjecture, and in light of Miao's result, one would like to consider extensions that fix the mean curvature of the boundary while, rather than fixing the first derivative of the metric. It would also be interesting if one can obtain Miao's result within this framework. However, it is not obvious how to develop the appropriate variational principle in this case; this is to be the subject of future work. 

In Section \ref{SPhaseSpace} we discuss the Hilbert manifold of extensions to be considered, which is essentially Bartnik's phase space with boundary conditions imposed. In Section \ref{SEnergy}, we introduce energy, momentum and mass definitions, and demonstrate how Corvino's result on static extensions can be obtained. Finally, in Section \ref{SConformal}, we consider a space of extensions in a prescribed conformal class. We give sufficient conditions to ensure that the infimum is realised within the fixed conformal class. However, as above, the boundary conditions considered here are not appropriate to be of direct significance to the Bartnik mass. It would be interesting to find a larger class of initial data for which a similar argument is possible, and impose Bartnik's geometric boundary conditions.

\section{The phase space}\label{SPhaseSpace}
Let $\mathcal{M}$ be a smooth asymptotically flat $3$-manifold with smooth boundary, $\Sigma$. We also assume that $\mathcal{M}$ has only a single asymptotic end; that is, there exists a compact set $K\supset\Sigma$ such that $\mathcal{M}\setminus K$ is diffeomorphic to $\R^3$ minus the closed unit ball, $\phi:\mathcal{M}\setminus K\rightarrow \R^3\setminus \overline{B_0(1)}$. On $\mathcal{M}\setminus K$ we define $\og$ to be the pullback of the Euclidean metric via $\phi$, and let $r$ be the Euclidean radial coordinate function composed with $\phi$. On $K$, $\og$ is extended to be smooth, bounded and positive definite, while $r$ is smooth and bounded between $\frac{1}{2}$ and $2$. Throughout, ``$\circ$" will indicate quantities defined with respect to the background metric $\og$. In order to include the asymptotics and prescribe the data on the boundary, we define the weighted Sobolev spaces, which are equipped with the following norms:

\begin{align}
\left\|u\right\|_{p,\delta}&=
\left\{
\begin{array}{ll}
\left(\int_\mathcal{M}\left| u\right|^p r^{-\delta p-3}d\mu_o\right)^{1/p},& p<\infty,\\
\text{ess sup}_\mathcal{M}(r^{-\delta}|u|), & p=\infty,
\end{array}
\right.
\\
\left\|u\right\|_{k,p,\delta}&=\sum_{j=0}^k\|\mathring{\nabla}^j u\|_{p,\delta-j}.
\end{align}

The spaces $L^p_\delta$ and $\W^{k,p}_{\delta}$ are defined as the completion of smooth, compactly supported functions on $\mathcal{M}\setminus\Sigma$ with respect to these norms. Spaces of sections of bundles are defined as usual and we use the standard notation $\W^{k,2}_{\delta}=\bH^k_{\delta}$. We also make use of the spaces $W^{k,p}_{\delta}$ and $H^k_\delta$, defined as the completion of smooth functions with bounded support on $\mathcal{M}$. That is, the overline indicates spaces of functions that vanish on the boundary, in the trace sense. 

Initial data for the Einstein equations is given by a Riemannian metric $g$ and a contravariant symmetric 2-tensor density $\pi$, on a 3-manifold $\mathcal{M}$. Motivated by the Bartnik mass, we are interested in the space of asymptotically flat extensions to a region $\Omega$ in a given initial data set $(\tilde{\mathcal{M}},\tilde{g},\tilde{\pi})$. In the context considered here, an extension to $\Omega$ is an asymptotically flat manifold $\mathcal M$ with boundary $\Sigma$ that may be identified with $\partial\Omega$ via a diffeomorphism, such that the initial data agrees across the boundary.

Let $\opi$ be some fixed symmetric 2-tensor density that is supported near $\Sigma$. We then consider a choice of $\og$, which we are free to specify near $\Sigma$, and $\opi$ as providing our boundary conditions; explicitly, we define the spaces
$$\mathcal{G}:=\{g\in S_2:g>0, g-\og\in \bH^2_{-1/2}\},\qquad\mathcal{K}:=\{\pi\in S^2\otimes\Lambda^3:\pi-\opi\in \bH^1_{-3/2}\},$$
where $\Lambda^3$ is the space of $3$-forms on $\mathcal{M}$, and $S_2$ and $S^2$ are symmetric covariant and contravariant tensors on $\mathcal{M}$ respectively. Initial data $(g,\pi)\in \mathcal{F}:=\mathcal{G}\times \mathcal{K}$ on $\mathcal M$ is to be thought of as an extension of some $\Omega$, where $\og$ and $\opi$ are given by extending $\tilde{g}$ and $\tilde{\pi}$ into $\mathcal M$; that is, $\Omega$ can be glued to $\mathcal{M}$ along the boundaries and data on $\Omega$ can be extended into $\mathcal{M}$. However, while this is the motivation for fixing the data on the boundary, we do not make reference to $\tilde{g}$ and $\tilde{\pi}$, as $\og$ and $\opi$ may be freely specified. Note that the space $\mathcal{F}$ imposes both the asymptotics and boundary conditions.

The constraint map, $\Phi:\mathcal{F}\rightarrow\mathcal{N}\subset L^2_{-5/2}(\Lambda^3\times T^*\mathcal{M}\otimes\Lambda^3)$, is given by
\begin{align}
\Phi_0(g,\pi)&=R(g)\sqrt{g}-(\pi^{ij}\pi_{ij}-\frac{1}{2}(\pi^k_k)^2)g^{-1/2},\\
\Phi_i(g,\pi)&=2\nabla_k\pi^k_i.
\end{align}
The constraint equations are then given by $\Phi(g,\pi)=(16\pi\rho,16\pi j_i)$, where $\rho$ and $j_i$ are the source energy and momentum densities respectively.

Bartnik's work on the phase space relies on the use of weighted Sobolev-type inequalities, most of which remain valid on an asymptotically flat manifold with boundary (see Theorem 1.2 of Ref. \cite{AF}), although some care should be taken with the use of the weighted Poincar\'e inequality. As such, it is straightforward to verify that the majority of Bartnik's proof, showing the level sets of $\Phi$ are Hilbert manifolds (cf. Conjecture \ref{conj}, below), is valid in the case where $\mathcal M$ has a boundary and the initial data has the boundary conditions imposed by $\mathcal{F}$. The place that Bartnik's proof, when applied to this case, breaks down is in proving the linearised constraint map is surjective. In fact, if $\mathcal{N}=L^2_{-5/2}(\Lambda^3\times T^*\mathcal{M}\otimes\Lambda^3)$ then this is almost certainly false. Nevertheless, we expect the following conjecture to be true and intend on pursuing a proof of this as part of future work.
\begin{conjecture}[cf. Theorem 3.12 of Ref. \cite{phasespace}]\label{conj}
For some $\mathcal{N}$, $\Phi$ is a smooth map of Hilbert manifolds and $D\Phi(g,\pi)$ is surjective at each point $(g,\pi)\in\mathcal{G}\times{K}$. It then follows from the implicit function theorem that the level sets of $\Phi$ are Hilbert submanifolds of $\mathcal{F}$.
\end{conjecture}
That is, the space of possible extensions to a given domain $\Omega$ is a Hilbert manifold; we refer to this as the constraint submanifold, and use the notation $\mathcal{C}(\rho,j)=\Phi^{-1}(16\pi\rho,16\pi j)$. It should be noted that such a result is required for the arguments outlined in Section \ref{SEnergy}.

\section{Static metric extensions}\label{SEnergy}
The total ADM energy and momentum are respectively given by
\begin{align}
16\pi E&:=\oint_{\infty}\og^{ik}(\onabla_k g_{ij}-\onabla_j g_{ik})dS^j,\label{usualE}\\
16\pi p_i&:=2\oint_{\infty}\pi_{ij}dS^j.\label{usualp}
\end{align}
Note that we omit reference to $\og$ in writing $dS$, as the definitions are independent of the asymptotically flat metric used to define the area measure (cf. Lemma \ref{lemma44} below, particularly its application in the proof of Proposition \ref{thmcorvino}). Note that we have also made a slight abuse of notation here, as $\pi$ is in fact a density. Often the quantity $E$ is called the mass, however we reserve the term mass for the quantity, $m=\sqrt{E^2-|p|^2}$; we assume the dominant energy condition here to ensure this is real.

We are now in a position to discuss critical points of the mass/energy over the space of extensions, and in particular show how Bartnik's work is easily adapted to give another proof of Corvino's result on static metric extensions. Previously, the author considered evolution exterior to a $2$-surface \cite{PhysRevD.90.104034}; however, the data was not fixed on the boundary so the conclusion is somewhat different. In the context of the static metric extension conjecture, it is more interesting to consider fixed boundary data.

It should be emphasised here that this section is to be understood as a discussion, or commentary, on the implications of Bartnik's earlier work, rather than a new independent work. It is our hope that Conjecture \ref{conj} is established in the near future, and furthermore that the boundary conditions can be replaced with Bartnik's geometric boundary conditions. A variational argument such as this, with boundary conditions directly relevant to the Bartnik mass, would be of significant interest.
\begin{proposition}[cf. Corollary 6.2 of Ref. \cite{phasespace}]\label{prop1}
Fix $(g,\pi)\in\mathcal{C}(\rho,j)$, where $(\rho,j)\in L^1$. Assume that Conjecture \ref{conj} holds, and further assume that a weak solution, $\lambda$ to $D\Phi(g,\pi)^*[\lambda]=f$ for $f\in L^2_{-5/2}(S^2\otimes\Lambda^3)\times H^{1}_{-3/2}(S_2)$, is indeed a strong solution; that is, if $\lambda\in\mathcal{N}^*$ satisfies
$$\int_{\mathcal{M}}\lambda\cdot D\Phi(g,\pi)[h,p]=\int_\mathcal{M}f\cdot (h,p)$$
for all $(h,p)\in T_{(g,\pi)}\mathcal{F}$, then $D\Phi(g,\pi)^*[\lambda]=f$. Then if $Dm(g,\pi)[h,p]=0$ for all $(h,p)\in T_{(g,\pi)}\mathcal{C}(\rho,j)$, $(g,\pi)$ is a stationary initial data set.
\end{proposition}
The assumption that weak solutions imply strong solutions is a result obtained in proving the analogue of Conjecture 1 in \cite{phasespace}; it is expected to be true here too. Note that Conjecture \ref{conj} and the additional condition regarding weak solutions are precisely the requirements for the arguments given in Sections 5 and 6 of \cite{phasespace} to hold. In fact, Proposition \ref{prop1} follows directly from these arguments if one replaces the function spaces used there with ours, which include boundary conditions. The proof of Proposition \ref{thmcorvino} is essentially the same argument; as such, we only present this and refer the reader to \cite{phasespace} for the proof of Proposition \ref{prop1}. Note that Proposition \ref{prop1} above differs from Corvino's static extension result, which in our framework is essentially Proposition \ref{thmcorvino} below. Let $\overline{R}(g)=R(g)\sqrt{g}$, and note that Conjecture \ref{conj} implies $D\overline{R}(g):T_g\mathcal{G}\rightarrow T_g\mathcal{N}_0$ is surjective, where $\mathcal{N}_0$ is the projection of $\mathcal{N}$ onto the first ($\Lambda^3$) factor. In the following, a static initial data metric is to be taken as a metric $g$ such that there exists a function $N$, asymptotic to a constant, satisfying $D\overline{R}(g)^*[N]=0$. A well-known result of Moncrief \cite{Moncrief1} (see also Ref. \cite{KIDs}) implies that the evolution of such an initial data metric is static. In fact, explicit calculation shows
$$ D\overline{R}(g)^*[N]=(\nabla^i\nabla^jN-\Delta_gN-NRic^{ij}+\frac{N}{2}Rg^{ij})\sqrt{g},$$
which reduces to the well-known static vacuum equations (cf. Ref. \cite{corvino2000}) when $R=0$. It is worth noting that such an $N$ is unique up to scaling, provided the initial data is not flat \cite{miaotam}. Stationary initial data in the context of Proposition \ref{prop1} is to be understood similarly.
\begin{proposition}\label{thmcorvino}
Fix $g\in\hat{\mathcal{C}}(\rho)=\{g\in \mathcal{G}: \overline{R}(g)=16\pi\rho\}$, where $\rho\in L^1$. Assume that Conjecture \ref{conj} holds, and further assume that a weak solution, $\lambda_0$ to $D\overline{R}(g)^*[\lambda_0]=f$ for $f\in L^2_{-5/2}(S^2\otimes\Lambda^3)$, is indeed a strong solution; that is, if $\lambda_0\in\mathcal{N}^*_0$ satisfies
$$\int_{\mathcal{M}}\lambda_0D\overline{R}(g)[h]=\int_\mathcal{M}f\cdot h$$
for all $h\in T_g\mathcal{G}$, then $D\overline{R}(g)^*[\lambda_0]=f$. Then if for all $h\in T_g\hat{\mathcal{C}}(\rho)$, we have $DE(g)[h]=0$, it follows that $g$ is a static initial data metric.
\end{proposition}
\begin{proof}

Fix some constant $N_\infty\in\R$, then for $N$ satisfying $(N-N_\infty)\in H^2_{-1/2}(\mathcal{M})$, consider the following Lagrange function for the Hamiltonian constraint:
\begin{equation}
L(g;N)=N_\infty E(g)-\int_\mathcal{M}NR(g)\sqrt{g}\label{modlagrange}.
\end{equation}
Note that this is essentially the Regge-Teitelboim Hamiltonian with the momentum set to zero. While this Lagrange function is well-defined on $\hat{\mathcal{C}}(\rho)$, it is not the case for a generic $g\in\mathcal{G}$. However, by writing the energy as the volume integral of a divergence over $\mathcal{M}$ (cf. Proposition 4.5 of \cite{phasespace}) the terms combine and the dominant terms cancel out. In particular, the regularised Lagrange function,
\begin{align}
L_{\text{reg}}(g;N)=&\int_\mathcal{M}(N_\infty-N)\overline{R}(g)\nonumber\\
&+\int_\mathcal{M}N_\infty(\og^{ik}\og^{jl}(\onabla_k\onabla_lg_{ij}-\onabla_i\onabla_kg_{jl})\sqrt{\og}-\overline{R}(g)),\label{Lreg}
\end{align}
is well defined on all of $\mathcal{G}$, and equal to $L(g;N)$ where the latter is defined. The first integral clearly converges since $\rho\in L^1$ and the second is bound by noting that the dominant term in $R(g)$ is $\og^{ik}\og^{jl}(\onabla_k\onabla_lg_{ij}-\onabla_i\onabla_kg_{jl})$, when expressed in terms of the background connection (see Proposition 4.2 of \cite{phasespace}, and the explicit expression for $R(g)$ can be found in Appendix A of \cite{MyThesis}).
Now we show that if $(N-N_\infty)\in W^{2,2}_{-1/2}$ then we have
\begin{equation}
DL_{\text{reg}}(g;N)[h]=-\int_\mathcal{M}h\cdot D\overline{R}(g)^*[N],\label{DLreg}
\end{equation}
where $D\overline{R}(g)^*$ is the formal adjoint of $D\overline{R}(g)$. This is easily seen by direct calculation, making use of the following Lemma.
\begin{lemma}[Lemma 4.4 of Ref. \cite{phasespace}]\label{lemma44}
Let $S_R$ be the Euclidean sphere of radius $R$, $E_R$ be the exterior region to $S_R$ -- the connected component of ${\mathcal{M}\setminus S_R}$ containing infinity -- and $A_R$ be the annular region between $S_R$ and $S_{2R}$. Suppose $u\in H^{1}_{-3/2}(E_{R_0})$, then for every $R\geq R_0$,
\begin{equation}
\oint_{S_R}|u|dS\leq c R^{1/2}\|u\|_{1,2,-3/2:A_R}\label{sphereesimate},
\end{equation}
where $c$ is independent of $R$.
\end{lemma}
 Note first
$$ h\cdot D\overline{R}(g)^*[N]-ND\overline{R}(g)[h]=\nabla^i(N(\onabla_ih^k_k-\nabla^jh_{ij})+h_{ij}\onabla^jN-h^k_k\onabla_iN)\sqrt{g},$$
and then taking the integral of this divergence over $\mathcal M$ results in several boundary terms; those on $\partial M$ vanish due to the boundary conditions. The boundary terms at infinity of the form $h\onabla N$ are $o(r^{-2})$, and controlled by
$$\|h\onabla N\|_{L^1(S_R)}\leq O(R^{1/2})\sup\limits_{S_R}|h|\|N\|_{2,2,-1/2}=o(1)$$
therefore the surface integrals at infinity also vanish. Now, by rewriting
$$\nabla^i(N(\onabla_ih^k_k-\nabla^jh_{ij}))=\nabla^i((N-N_\infty)(\onabla_ih^k_k-\nabla^jh_{ij})+N_\infty(\onabla_ih^k_k-\nabla^jh_{ij}))$$
we see that the integral of the first term again vanishes, since
$$\|(N-N_\infty)\onabla h\|_{L^1(S_R)}\leq O(R^{1/2})\sup\limits_{S_R}|N-N_\infty|\|h\|_{2,2,-1/2}=o(1).$$
We are therefore left with
$$\int_{\mathcal{M}} (h\cdot D\overline{R}(g)^*[N]-ND\overline{R}(g)[h])=\int_M(N_\infty\nabla^i(\onabla_ih^k_k-\nabla^jh_{ij}))\sqrt{g}.$$
By making similar use of $\onabla-\nabla$ and $\sqrt{g}-\sqrt{\og}$, we establish (\ref{DLreg}), which is valid for all $h\in T_g\mathcal{G}$.

We now employ the following theorem of Lagrange multipliers for Banach manifolds  (see Theorem 6.3 of \cite{phasespace}).
\begin{theorem}\label{banach}
Suppose $K:B_1\rightarrow B_2$ is a $C^1$ map between Banach manifolds, such that $DK(u):T_uB_1\rightarrow T_{K(u)}B_2$ is surjective, with closed kernel and closed complementary subspace for all $u\in K^{-1}(0)$. Let $f\in C^1(B_1)$ and fix ${u\in K^{-1}(0)}$, then the following statements are equivalent:
\begin{enumerate}[(i)]
\item For all $v\in\ker DK_u$, we have
\begin{equation}Df(u)[v]=0.\end{equation}
\item There is $\lambda\in (T_{K(u)}B_2)^*$ such that for all $v\in T_u B_1$,
\begin{equation}Df(u)[v]=\left<\lambda,DK(u)[v]\right>,\end{equation}	
where $\left< \, , \right>$ refers to the natural dual pairing.
\end{enumerate}
\end{theorem}
Let $K(g)=\overline{R}(g)-16\pi\rho$, so that $\hat{\mathcal{C}}(\rho)=K^{-1}(0)$ and $T_g\hat{\mathcal{C}}(\rho)=\ker(DK(g))$, and let $f(g)=L_{\text{reg}}(g;N)$.

Then if $g$ is a critical point of the ADM energy over the space of extensions satisfying $\overline{R}(g)=16\pi\rho$, we have $DE(g)[h]=0$ for all $h\in\ker(DK(g))$; that is, $(i)$ in the above theorem is satisfied. It follows that there exists $\lambda\in (T_{K(g)}\mathcal{N}_0)^*$, such that
$$ -\int_\mathcal{M}h\cdot D\overline{R}(g)^*[N] =Df(g)[h]=\int_\mathcal{M} \lambda D\overline{R}(g)[h]$$
for all $h\in T_g\mathcal{G}$. That is, $\lambda$ is a weak solution to $D\overline{R}(g)^*[\lambda]=F$, where $F=-D\overline{R}(g)^*[N]\in L^2_{-5/2}(S^2\otimes\Lambda^3)$. By assumption, this is a strong solution and we therefore have
$$ D\overline{R}(g)^*[\lambda+N]=0;$$
that is, $g$ is a static initial data set, with static potential $(\lambda+N)\rightarrow N_\infty$ at infinity.
\end{proof}
It is clear that Theorem \ref{banach} should imply a converse statement, however little can be said about this without explicitly knowing $\mathcal{N}$.
\begin{remark}
If $g\in C^3$ then an argument of Fischer-Marsden \cite{fischer1975} (cf. Ref. \cite{corvino2000}, Proposition 2.3) says staticity implies $R(g)=0$. In particular, the above result then implies for $\rho\neq0$, any critical points of the mass (subject to the hypotheses holding) should be rougher than $C^3$. If the condition $g\in C^3$ can be weakened to $g\in H^2_{loc}$, then one concludes that there are no critical points of the mass for $\rho\neq0$.
\end{remark}

\section{Energy minimisers in a fixed conformal class}\label{SConformal}
A standard approach to simplify the constraint equations is to look for solutions within a fixed conformal class (see Ref. \cite{BaIsConstraints2004} and references therein); in this case, the Hamiltonian constraint becomes elliptic. Here we make use of this simplification by considering the space of extensions to $\Omega$ within a given conformal class. Specifically, we consider a fixed metric $\tg\in\mathcal{G}$ and consider extensions of the form ${g(\phi)=e^{4\phi}\tilde{g}}$, with $\phi\in\bH^2_{-1/2}$. For simplicity, we assume that $\mathcal{M}$ is diffeomorphic to $\R^3\setminus\overline{B_0(1)}$; that is, we consider the most natural extensions to $\Omega$. This affords us the use of the weighted Poincar\'e inequality (see, for example, Lemma 3.10 of Ref. \cite{phasespace}).

It should be emphasised that the boundary conditions imposed by the condition $\phi\in\bH^2_{-1/2}$ are too strong to be of direct significance to the Bartnik mass. While motivated by the Bartnik mass, the results in this section are simply of mathematical interest; it is the hope that similar ideas can be used to prove the existence of a minimiser in a much larger class of extensions, and therefore gain insight into the computability of the Bartnik mass. One natural candidate for a larger class of extensions would be to consider Brill initial data, using a variation of Dain's mass functional (see, for example, \cite{dain1}).

The scalar curvature of $g=e^{4\phi}\tilde g$ is given by the well-known formula,
\begin{equation*}
R(g)=e^{-4\phi}(\tilde{R}-8|\tilde{\nabla}\phi|^2-8\tilde{\Delta}\phi),
\end{equation*}
where $\sim$ indicates quantities defined with respect to $\tilde{g}$. This allows us to write the conformal constraint map, $\hat{\Phi}:\bH^2_{-1/2}(\mathcal{M})\times\mathcal{K}\rightarrow L^2_{-5/2}(\Lambda^3\times T^*\mathcal{M}\otimes\Lambda^3)$, as
\begin{align}
\hat{\Phi}_0(\phi,\pi)&=e^{2\phi}\left[(\tilde{R}-8|\tilde{\nabla}\phi|^2-8\tilde{\Delta}\phi)\sqrt{\tilde{g}}-\tilde{g}_{ik}\tilde{g}_{jl}(\pi^{ij}\pi^{kl}-\frac{1}{2}\pi^{ik}\pi^{jl})\tilde{g}^{-1/2}\right],\label{phihamconst}\\
\hat{\Phi}_i(\phi,\pi)&=2e^{4\phi}\left(\tilde{g}_{ip}\tilde{\nabla}_k\pi^{kp}+4\tg_{ip}\pi^{kp}\tnabla_k\phi-2\tilde{g}_{jp}\pi^{jp}\tnabla_i\phi\right).
\end{align}
From this point on, we will raise and lower indices, and consider the weighted Sobolev norms, with respect to $\tilde{g}$ rather than $g$ or $\og$. Note that the domain of $\hat{\Phi}$ enforces the boundary conditions on $(g,\pi)$; in particular, the conformal metric $\tg$ must itself be an extension of $\Omega$, although it need not necessarily satisfy the constraints.
\begin{proposition}
For any $\phi\in\bH^2_{-1/2}$, we have $g=e^{4\phi}\tg\in\mathcal{G}$.
\end{proposition}
\begin{proof}
It is clear that $e^{4\phi}\tg$ is positive-definite, and using the standard weighted Sobolev-type inequalities we have,
\begin{align*}
\|e^{4\phi}\tg-\tg\|_{2,2,-1/2}\leq & \,c\|\tg\|_{\infty,0}(\|e^{4\phi}-1\|_{2,-1/2}+\|e^{4\phi}\tnabla\phi\|_{2,-3/2}\\
&+\|e^{4\phi}\tnabla^2\phi\|_{2,-5/2})\\
\leq& c\,\|\tg\|_{\infty,0}(\|e^{4\phi}-1\|_{2,-1/2}+\|e^{4\phi}\|_{\infty,0}\|\tnabla\phi\|_{1,2,-3/2}).
\end{align*}
Note that $\phi$ is continuous by the Morrey embedding and $|e^{4\phi}-1|<5|\phi|$ near infinity, so $\|e^{4\phi}-1\|_{2,-1/2}<\infty$.
\end{proof}

\begin{proposition}\label{propEformula}
Assume $(\phi,\pi)\in \bH^2_{-1/2}(\mathcal{M})\times\mathcal{K}$ satisfies $\hat{\Phi}_0(\phi,\pi)=16\pi \rho$, where ${\rho\in L^1_{-3}(\Lambda^3(\mathcal{M}))}$ is the source energy density. The ADM energy can then be expressed as,
\begin{equation}
16\pi E=16\pi\tilde{E}+\int_\mathcal{M}\big{(}(8|\tnabla\phi|^2-\tilde{R})\sqrt{\tilde{g}}+(\pi^{ij}\pi_{ij}-\frac{1}{2}(\pi^k_k)^2)/\sqrt{\tilde{g}}+16\pi e^{-2\phi}\rho\big{)},\label{newEdefn}
\end{equation}
where $\tilde{E}$ is the ADM energy of $\tilde{g}$.
\end{proposition}
\begin{proof}
First we write $E$ in terms of $\phi$ and $\tg$,
\begin{align}
16\pi E=&\,\oint_\infty\og^{ik}e^{4\phi}\left(4\onabla_k(\phi)\tilde{g}_{ij}+\onabla_k\tilde{g}_{ij}-4\onabla_j(\phi)\tilde{g}_{ik}-\onabla_j\tilde{g}_{ik}\right)dS^j\nonumber\\
=&\,\oint_\infty\og^{ik}\left(4\onabla_k(\phi)\tilde{g}_{ij}+\onabla_k\tilde{g}_{ij}-4\onabla_j(\phi)\tilde{g}_{ik}-\onabla_j\tilde{g}_{ik}\right)dS^j\nonumber\\
&+\oint_\infty\og^{ik}(e^{4\phi}-1)\left(4\onabla_k(\phi)\tilde{g}_{ij}+\onabla_k\tilde{g}_{ij}-4\onabla_j(\phi)\tilde{g}_{ik}-\onabla_j\tilde{g}_{ik}\right)dS^j\label{eqEphidefn}.
\end{align}
Lemma \ref{lemma44} can now be used to control the second integrand in Eq. (\ref{eqEphidefn}),
\begin{align*}
\Big{|}\oint_{S_R}\og^{ik}&(e^{4\phi}-1)\left(4\onabla_k(\phi)\tilde{g}_{ij}+\onabla_k\tilde{g}_{ij}-4\onabla_j(\phi)\tilde{g}_{ik}-\onabla_j\tilde{g}_{ik}\right)dS^j\Big{|} \\
&\leq c\|e^{4\phi}-1\|_{\infty:S_R}(\|\tg\|_{\infty:S_R}\|\onabla\phi\|_{1:S_R}+\|\onabla \tg\|_{1:S_R})\\
&\leq O(R^{1/2})\|e^{4\phi}-1\|_{\infty:S_R}(\|\tg\|_{\infty:S_R}\|\onabla\phi\|_{1,2,-3/2}+\|\onabla \tg\|_{1,2,-3/2}).
\end{align*}
Now making use of the continuity and asymptotics of $e^4\phi$ and $\tg$, the right-hand-side simply becomes $o(1)$ and therefore vanishes as $R$ tends to infinity. Eq. (\ref{eqEphidefn}) now becomes
$$16\pi E=\oint_\infty\og^{ik}\left(4\onabla_k(\phi)\tilde{g}_{ij}+\onabla_k\tilde{g}_{ij}-4\onabla_j(\phi)\tilde{g}_{ik}-\onabla_j\tilde{g}_{ik}\right)dS^j,$$
which can be expressed in terms of the energy, $\tilde{E}$, of $\tg$,
$$16\pi E=16\pi\tilde{E}+4\oint_\infty\og^{ik}\left(\onabla_k(\phi)\tilde{g}_{ij}-\onabla_j(\phi)\tilde{g}_{ik}\right)dS^j.$$
Since $(\og-\tg)\in\bH^2_{-1/2}$ and $\onabla\phi=\partial\phi=\tnabla\phi$, Lemma \ref{lemma44} can again be used to conclude
\begin{align*}
16\pi E&=16\pi\tilde{E}+4\oint_\infty\tg^{ik}\left(\tnabla_k(\phi)\tilde{g}_{ij}-\tnabla_j(\phi)\tilde{g}_{ik}\right)dS^j\\
&=16\pi\tilde{E}-8\oint_\infty\tnabla_j\phi dS^j.
\end{align*}

It is now simply a matter of applying the divergence theorem and making use of the Hamiltonian constraint (\ref{phihamconst}) to complete the proof.
 \end{proof}
Henceforth, when we write $E(\phi,\pi)$, we mean to take (\ref{newEdefn}) to be the definition of the energy, which is well-defined provided both $\tilde{R}$ and the source are integrable.

In the vacuum case ($\rho=0$), if $\tilde{R}=0$ then it is clear from (\ref{newEdefn}) that the energy of any solution $g$ in the conformal class of $\tilde{g}$ has energy greater than $\tilde{E}$, with equality only if $g=\tilde{g}$. That is, if there exists a metric $\hat{g}$ in the conformal class of $\tilde{g}$ with $R(\hat{g})=0$ and appropriate boundary conditions satisfied, then the infimum of the energy is attained by $\hat{g}$. Generically such a scalar-flat extension does not exist though, as our boundary conditions on $\tnabla\phi$ are too strong to ensure this. An argument of Cantor and Brill \cite{CantorBrill1981} proves the existence of scalar-flat metrics when no boundary is present, and a similar argument can be used to guarantee the existence of such an extension under Dirichlet boundary conditions\footnote{To the best of the author's knowledge, this argument hasn't been explicitly published; however, it is almost certainly true and we intend on explicitly verifying this as part of another project.}; however, this argument does not hold for the (stronger) boundary conditions here. 

Note that the proof of Proposition \ref{propEformula} requires the vanishing of $\tnabla g$ on $\Sigma$; if we relax the condition on $\tnabla \phi$, there is an extra surface integral on $\Sigma$ corresponding to the difference in the mean curvatures due to $g$ and $\tilde{g}$, reminiscent of the Brown-York quasilocal mass \cite{brownyork}. Interestingly, imposing Bartnik's geometric boundary conditions enforce that this term indeed vanishes.

Now we will need the following estimate for the proof of the main result of this section.
\begin{proposition}\label{propdelta}
For $u\in W^{1,2}_{\delta}$ and $\epsilon>0$, it holds that
\begin{equation}
\|u\|_{4,\delta}\leq c(\epsilon)\|u\|_{2,\delta}+\epsilon\|u\|_{1,2,\delta}.
\end{equation}
\end{proposition}
\begin{proof}
This follows from the weighted H\"older and Sobolev inequalities, the definition of the weighted norms, and Young's inequality:
\begin{align*}
\|u\|_{4,\delta}&=\|u^{1/4}u^{3/4}\|_{4,\delta}\\
&\leq \|u^{1/4}\|_{8,\delta/4}\|u^{3/4}\|_{8,3\delta/4}\\
&= \|u\|^{1/4}_{2,\delta}\|u\|^{3/4}_{6,\delta}\\
&\leq c(\epsilon)\|u\|_{2,\delta}+\epsilon\|u\|_{6,\delta}\\
&\leq c(\epsilon)\|u\|_{2,\delta}+\epsilon\|u\|_{1,2,\delta}.
\end{align*}
 \end{proof}
The main theorem is divided into the two following, related statements:
\begin{theorem}\label{thmconf1}
Let $S^+_\alpha$ be the set of ${(\phi,\pi)\in \bH^2_{-1/2}(\mathcal{M})\times\mathcal{K}}$, satisfying the following conditions:
\begin{enumerate}[(i)]
\item $\hat{\Phi}_0(\phi,\pi)\geq0$,
\item $\phi\geq-\alpha$,
\item $\hat{\Phi}_0(\phi,\pi)\in L^1_{-3}$.
\end{enumerate}
Then either the infimum is achieved over $S^+_\alpha$, or for all minimising sequences $(\phi_n,\pi_n)\in S^+_\alpha$, that is sequences satisfying ${\lim\limits_{n\rightarrow\infty}E(\phi_n,\pi_n)=\inf\limits_{(\phi,\pi)\in S^+_\alpha}E(\phi,\pi)}$, we have that $$\max\{\|\hat{\Phi}(\phi_n,\pi_n)\|_{2,-5/2-\epsilon},\frac{\|\tnabla_k\pi^{ij}_n\|_{2,-5/2-\epsilon}}{\|\tnabla_j\pi^{ij}_n\|_{2,-5/2-\epsilon}}\}\rightarrow\infty$$
for all $\epsilon\in(0,\frac{1}{2})$.
\end{theorem}

\begin{theorem}\label{thmconf2}
Let $S^0=\{(\phi,\pi)\in \bH^2_{-1/2}(\mathcal{M})\times\mathcal{K}:\hat{\Phi}(\phi,\pi)=0\}$. Then either the infimum is achieved over $S^0$, or for all minimising sequences $(\phi_n,\pi_n)\in S^0$, we have that $\frac{\|\tnabla_k\pi^{ij}_n\|_{2,-5/2-\epsilon}}{\|\tnabla_j\pi^{ij}_n\|_{2,-5/2-\epsilon}}\rightarrow\infty$ for all $\epsilon\in(0,\frac{1}{2})$.
\end{theorem}
\begin{remark}
The conditions $(i)$ and $(iii)$ simply state the source energy density is non-negative and the total energy is finite, while condition $(ii)$ prevents the limiting metric from becoming degenerate. In Theorem \ref{thmconf1}, the alternative $\|\hat{\Phi}(\phi_n,\pi_n)\|_{2,-5/2-\epsilon}\rightarrow\infty$ simply excludes the possibility that source energy-momentum blows up as the ADM energy is minimised, which one imagines is certainly never the case for an initially integrable source. The alternative, $\frac{\|\tnabla_k\pi^{ij}_n\|_{2,-5/2-\epsilon}}{\|\tnabla_j\pi^{ij}_n\|_{2,-5/2-\epsilon}}\rightarrow\infty$, unfortunately doesn't have an obvious physical interpretation.
\end{remark}
Since the two theorems are similar, we prove them simultaneously, noting the relevant differences.
\begin{proof}

From Proposition \ref{propEformula} we have
\begin{align}
\|\tnabla\phi\|^2_{2,-3/2}+\|\pi\|^2_{2,-3/2}\leq &\, 32\pi(E-\tilde{E})+2\int_\mathcal{M}\tilde{R}\sqrt{\tilde{g}}\nonumber\\
&-2\pi\int_\mathcal{M}e^{-2\phi}\hat{\Phi}_0(\phi,\pi)\nonumber\\
\leq &\, 32\pi E+\tilde{C}.\label{Elowerbound}
\end{align}

This implies that if the initial data is sufficiently large then we can guarantee that the energy is large. Let $S$ be either of the sets $S^+_\alpha$ or $S^0$, and define $E_0=\inf_{(\phi,\pi)\in S} E(\phi,\pi)$. Now let $(\phi_n,\pi_n)$ be a sequence in the constraint set such that $E(\phi_n,\pi_n)\rightarrow E_0$. Note that (\ref{Elowerbound}) and the Poincar\'e inequality imply that there exists a constant $K$ such that for $\|(\phi,\pi)\|_{H^1_{-1/2}\times L^2_{-3/2}}>K$, we have $E(\phi,\pi)>E_0+1$. That is, truncating the beginning of the sequence if necessary, ${\|(\phi_n,\pi_n)\|_{H^1_{-1/2}\times L^2_{-3/2}}<K}$. In particular, extracting a subsequence if required, $(\phi_n,\pi_n)$ convergences weakly in $H^1_{-1/2}\times L^2_{-3/2}$ to a limit, $(\phi_\infty,\pi_\infty)$. It remains to be shown that $(\phi_\infty,\pi_\infty)\in S$.

In the following, it will be convenient to let $\dn=-\epsilon/2$, then we assume that $\max\{\|\hat{\Phi}(\phi_n,\pi_n)\|_{2,-5/2+2\dn},\frac{\|\tnabla_k\pi^{ij}_n\|_{2,-5/2+2\dn}}{\|\tnabla_j\pi^{ij}_n\|_{2,-5/2+2\dn}}\}<C$, and prove below that the infimum is realised in $S$.

Proposition \ref{propdelta} and the definition of $\hat{\Phi}_0$ give
\begin{align}
\|\tilde{\Delta}\phi_n\|_{2,-5/2+2\dn}\leq &\, c(\|\tilde{R}\|_{2,-5/2+2\dn}+\|\tnabla\phi_n\|^2_{4,-5/4+\dn}+\|\pi_n\|^2_{4,-5/4+\dn}\nonumber\\
&+\|e^{-2\phi_n}\Phi_0(\phi_n,\pi_n)\|_{2,-5/2+2\dn})\nonumber\\
\leq &\,c(\epsilon)(1+\|\pi_n\|_{2,-5/4+\dn}^2+\|\tnabla\phi_n\|^2_{2,-5/4+\dn})\nonumber\\
&+\epsilon(\|\pi_n\|^2_{1,2,-5/4+\dn}+\|\tnabla\phi_n\|^2_{1,2,-5/4+\dn}),\label{laplacephi}
\end{align}
which follows from the assumption $\|\hat{\Phi}(\phi_n,\pi_n)\|_{2,-5/2+2\dn}<C$ and condition $(ii)$ for the proof of Theorem \ref{thmconf1}, and from $\hat{\Phi}(\phi_n,\pi_n)=0$ for the proof of Theorem \ref{thmconf2}.

Similarly, the assumption $\frac{\|\tnabla_k\pi^{ij}_n\|_{2,-5/2+2\dn}}{\|\tnabla_j\pi^{ij}_n\|_{2,-5/2+2\dn}}<C$ and the definition of $\Phi_i$ gives\begin{align}
\|\tnabla\pi_n\|_{2,-5/2+2\dn}\leq &\,c (\|\tnabla\phi_n\|_{4,-5/4+\dn}   \|\pi_n\|_{4,-5/4+\dn}\nonumber\\
&+\|e^{-4\phi_n}\Phi_i(\phi_n,\pi_n)\|_{2,-5/2+2\dn})\nonumber\\
\leq &\,c (\|\tnabla\phi_n\|_{4,-5/4+\dn}^2  + \|\pi_n\|_{4,-5/4+\dn}^2+1)\nonumber\\
\leq &\,c(\epsilon)(\|\tnabla\phi_n\|_{2,-5/4+\dn}^2+\|\pi_n\|_{2,-5/4+\dn}^2+1)\nonumber\\
&+\epsilon(\|\tnabla\phi_n\|_{1,2,-5/4+\dn}^2+\|\pi_n\|_{1,2,-5/4+\dn}^2)\label{piestimate}.
\end{align}

We now recall the scale-broken estimate (Theorem 1.10 of Ref. \cite{AF}, Proposition 4.13 of Ref. \cite{MaxwellThesis}):
\begin{equation}
\|u\|_{2,2,\delta}\leq C\left(\|\tilde{\Delta}u\|_{2,\delta-2}+\|u\|_{2,0}\right).\label{scalebrokenest}
\end{equation}
Note that the application of the scale-broken estimate here requires $\epsilon\neq\frac12$. Combining (\ref{scalebrokenest}) with (\ref{laplacephi}), applying the weighted Poincar\'e inequality (cf. Lemma 3.10 of Ref. \cite{phasespace}) to (\ref{piestimate}), and choosing $\epsilon$ sufficiently small gives
\begin{align*}
\|\phi_n\|_{2,2,-1/2+2\dn}+\|\pi_n\|_{1,2,-3/2+2\dn}&\leq c \left(1+\|\phi_n\|^2_{1,2,-1/4+\dn}+\|\pi_n\|^2_{2,-5/4+\dn}\right)\\
&\leq c \left(1+\|\phi_n\|^2_{1,2,-1/2}+\|\pi_n\|^2_{2,-3/2}\right),
\end{align*}
since $\dn>-\frac14$. Weak convergence in ${H^2_{-1/2+2\dn}\times H^1_{-3/2+2\dn}}$ follows, and since $\dn<0$, the weighted Rellich compactness theorem (Lemma 2.1 of Ref \cite{ellipticsys}) implies strong convergence in $H^1_{-1/2}\times L^2_{-3/2}$.

At this point we consider $S=S^+_\alpha$ explicitly, and demonstrate that if $(\phi_n,\pi_n)\in S^+_\alpha$, then $(\phi_\infty,\pi_\infty)$ also satisfies conditions $(i)-(iii)$. Consider 
$$F_n=(\tilde{R}-8|\tilde{\nabla}\phi_n|^2-8\tilde{\Delta}\phi_n)\sqrt{\tilde{g}}-(\pi_n^2-\frac{1}{2}(\text{tr}_{\tg}\pi_n)^2)\tilde{g}^{-1/2}.$$
Note that the $|\tnabla\phi_n|^2$ and $\pi_n^2$ terms converge weakly in $L^2_{-5/2}$ since $${\|\pi^2\|_{2,-5/2}=\|\pi\|_{4.-5/4}^2\leq C\|\pi\|_{1,2,-3/2}^2}.$$ Furthermore, as the map $\pi\mapsto\pi^2$ is a bounded polynomial function from $L^2_{-3/2}$ to $L^1_{-3}$, it is smooth (see, for example, Chapter 26 of \cite{hillephillips}); that is, $\pi^2_n$ converges to $\pi^2_\infty$ strongly in $L^1_{-3}$ and by uniqueness of limits $\pi_n^2$ converges weakly in $L^2_{-5/2}$ to $\pi^2_\infty$. Note that $|\tnabla\phi_n|^2$ is similar. By simply integrating $\tilde{\Delta}\phi_n$ against a test function and integrating by parts, it is clear $\tilde{\Delta}\phi_n$ converges to $\tilde{\Delta}\phi_\infty$ weakly in $L^2_{-5/2}$. It follows that $F_n$ converges weakly in $L^2_{-5/2}$ to
$$F_\infty=(\tilde{R}-8|\tilde{\nabla}\phi_\infty|^2-8\tilde{\Delta}\phi_\infty)\sqrt{\tilde{g}}-(\pi_\infty^2-\frac{1}{2}(\text{tr}_{\tg}\pi_\infty)^2)\tilde{g}^{-1/2}.$$
We prove $F_\infty\geq0$ by contradiction; assume there is a bounded set $U\in\mathcal{M}$ such that $F_\infty<0$ on $U$. Let $\chi_U$ be the characteristic function of $U$, then by the weak convergence of $F_n$ we have
$$\int_U F_n=\int_\mathcal{M}F_n\chi_U\rightarrow\int_\mathcal{M}F\chi_U=\int_U F_\infty.$$
Since $F_n\geq0$ by assumption, we have a contradiction and it therefore follows that $\hat{\Phi}_0(\phi_\infty,\pi_\infty)\geq0$. An almost identical, albeit simpler, argument shows $\phi_\infty\geq-\alpha$, and from the definition of $E$ it is obvious that $\int\hat{\Phi}(\phi_\infty,\pi_\infty)<\infty$. We therefore conclude ${(\phi_\infty,\pi_\infty)\in S^+_\alpha}$.

The case $S=S^0$ is similar, albeit much simpler.

 \end{proof}

\begin{remark}
Theorem \ref{thmconf1} still holds without the assumption of condition $(i)$, however it is more interesting to impose the weak energy condition.
\end{remark}

\section{acknowledgements}
The author gratefully acknowledges many useful comments from two anonymous reviewers.

\bibliographystyle{spmpsci}
\bibliography{refs}

\end{document}